\documentclass[11pt,letterpaper]{article}

\usepackage{graphicx}
\usepackage{geometry}
\usepackage{amsfonts,amsmath,amssymb,amsthm}
\usepackage{bbm}
\usepackage[square]{natbib}
\usepackage{bm}
\usepackage[dvipsnames]{xcolor}
\usepackage[hyphens]{url}
\usepackage{hyperref}
\usepackage{cleveref}
\hypersetup{
	colorlinks=true,
	linkcolor=Blue,
	citecolor=BrickRed,
	urlcolor=PineGreen
}

\usepackage{pgfplots}
 
\pgfplotsset{compat = newest}

\DeclareMathOperator*{\argmax}{arg\,max}

\DeclareMathOperator*{\E}{\mathbb{E}}
\renewcommand{\d}{\,\mathrm{d}}

\newtheorem{theorem}{Theorem}[section]
\newtheorem{lemma}{Lemma}[section]
\theoremstyle{definition}
\newtheorem{remark}{Remark}
\newtheorem{proposition}{Proposition}[section]
\newtheorem{definition}{Definition}[section]
\newtheorem{problem}{Problem}[section]

\allowdisplaybreaks

\title{Improved Approximation to First-Best Gains-from-Trade}

\author{
	Yumou Fei\thanks{Peking University. Email: feiym2002@stu.pku.edu.cn}
}
\date{}

\begin{document}
	
\maketitle

\begin{abstract}
We study the two-agent single-item bilateral trade. Ideally, the trade should happen whenever the buyer's value for the item exceeds the seller's cost. However, the classical result of Myerson and Satterthwaite showed that no mechanism can achieve this without violating one of the Bayesian incentive compatibility, individual rationality and weakly balanced budget conditions. This motivates the study of approximating the trade-whenever-socially-beneficial mechanism, in terms of the expected gains-from-trade. Recently, Deng, Mao, Sivan, and Wang showed that the random-offerer mechanism achieves at least a 1/8.23 approximation. We improve this lower bound to 1/3.15 in this paper. We also determine the exact worst-case approximation ratio of the seller-pricing mechanism assuming the distribution of the buyer's value satisfies the monotone hazard rate property.
\end{abstract}

\section{Introduction}

Two-sided markets, with strategic players on both the sell-side and the buy-side, has been an important research topic in economics. This paper considers the simplest model of such market, the two-agent single-item bilateral trade.

\subsection{Model}
Suppose there is a single seller and a single buyer on the market. There is an item held by the seller and to be sold to the buyer. The buyer's private value for the item is a random variable $v$ with cumulative distribution function (CDF) $F$, and the seller's private cost of selling it out is a random variable $c$ with CDF $G$. We assume that $v$ and $c$ are independent, and that $F$ and $G$ have finite first moments. We need to design a mechanism which, on input $v$ and $c$, decides whether the transaction should happen. Let $x(\cdot,\cdot)$ be a (Borel-measurable) function from $\mathbb{R}^{2}$ to $[0,1]$, which denotes the trading probability decided by the mechanism, or the ``allocation rule''. The \textit{gains-from-trade} ($\mathsf{GFT}$), or the expected social utility gain from trading, is defined as 
$$\mathsf{GFT}:=\underset{\substack{v\sim F\\ c\sim G}}{\mathbb{E}}\left[(v-c)\cdot x(v,c)\right].$$
Thinking of $F$ and $G$ as given, what choice of the function $x$ maximizes $\mathsf{GFT}$? It is clear that whatever $F$ and $G$ are, the quantity $\mathsf{GFT}$ is always maximized at $x^{*}(v,c):=\mathbbm{1}\{v\geq c\}$. However, in reality, we have to take into consideration the strategic behavior of both sides of the market. As in many mechanism design problems, in order to make an allocation rule Bayesian incentive compatible (BIC) and individually rational (IR), we need to include a ``money transfer'' rule. But unlike in the case of one-sided auctions, the ``payment'' from the sell-side can be negative in two-sided markets. Thus, in addition to the commonly studied BIC and IR conditions, the ``balance of budget'' is another important consideration in the design of bilateral trade mechanisms. In particular, we are concerned with the following requirement: 
\begin{definition}
If the payment from the buyer is always at least the revenue of the seller, we say the mechanism is weakly budget balanced (WBB).  
\end{definition}

\cite{MS:1983} show that for very general $F$ and $G$, the ideal allocation rule $x^{*}(v,c)=\mathbbm{1}\{v\geq c\}$ cannot be made into a BIC, IR and WBB mechanism. But on the other hand, this idealism does offer a benchmark which we can try to approximate using mechanisms with these properties. We call the gains-from-trade achieved by $x^{*}(v,c)$ the \textit{first-best} $\mathsf{GFT}$. Namely, define
$$\mathsf{FB}:=\E_{\substack{v\sim F\\ c\sim G}}\left[(v-c)\cdot \mathbbm{1}\{v\geq c\}\right].$$

The maximum $\mathsf{GFT}$ achievable by Bayesian incentive compatible, individually rational and weakly budget balanced mechanisms is denoted $\mathsf{SB}$, or the \textit{second-best}. \cite{MS:1983} also describe a BIC, IR and WBB mechanism that actually achieves the second-best gains-from-trade. However, this second-best mechanism is complicated and difficult to implement in practice. As a consequence, the understanding of the second-best mechanism and the search for simple and practical alternatives have become major research problems. The following are some of the most studied simple mechanisms:

\begin{itemize}
\item The fixed-price mechanism. A fixed price $p$ is set and the transaction happens iff $v\geq p\geq c$. Formally, the allocation rule is $x(v,c)=\mathbbm{1}\{v\geq p\geq c\}$ and the buyer pays the seller the price $p$. The resulting gains-from-trade is 
$$\mathsf{FixedP}=\sup_{p\in\mathbb{R}}\int_{-\infty}^{p}\int_{p}^{+\infty}(v-c)\d F(v)\d G(c).$$ 
The fixed price mechanism has the additional advantage that it's dominant strategy incentive compatible (DSIC).
\item
The seller-pricing mechanism. The seller gets the right to set the price and the buyer can only decide whether or not to buy the item and pay this price. The buyer chooses to buy it iff $v\geq p$. Knowing his private cost $c$, the seller sets the price $p$ that maximizes $(p-c)(1-F(p^{-}))$, where $F(p^{-}):=\lim_{x\rightarrow p^{-}}F(x)$. The resulting gains-from-trade is 
$$\mathsf{SellerP}=\int_{-\infty}^{+\infty}\int_{p_{c}}^{\infty}(v-c)\d F(v)\d G(c),\quad\text{where }p_{c}\in\argmax_{p}(p-c)(1-F(p^{-})).$$
\item The buyer-pricing mechanism. The buyer gets the right to set the price and the seller can only decide whether or not to sell the item and take this price. The seller chooses to sell it iff $c\leq p$. Knowing his private value $v$, the buyer sets the price $p$ that maximizes $(v-p)G(p)$. The resulting gains-from-trade is 
$$\mathsf{BuyerP}=\int_{-\infty}^{+\infty}\int_{-\infty}^{p_{v}}(v-c)\d G(c)\d F(v),\quad\text{where }p_{v}\in\argmax_{p}(v-p)G(p).$$
\item The random-offerer mechanism. Flip a fair coin to decide who gets to set the price. This is a 50:50 mixture of the seller pricing mechanism and the buyer price mechanism, which extracts a gains-from-trade of
$$\mathsf{RandOff}=\frac{1}{2}\mathsf{SellerP}+\frac{1}{2}\mathsf{BuyerP}.$$
\end{itemize}
\begin{remark}\label{rem:symmetry}
It's well-known that the quantities $\mathsf{BuyerP}$ and $\mathsf{SellerP}$ are symmetric to each other. In fact, if the buyer's private value suddenly became the random variable $-c$ and the seller's private cost became the random variable $-v$, then the two quantities would swap with each other. As a consequence, in terms of the gains-from-trade, there is a good symmetry between the seller-pricing mechanism and the buyer-pricing mechanism.
\end{remark}

\subsection{Our Results}\label{subsec:results}
In this paper, we focus on the following problems:
\begin{problem}\label{prob:basic}
Among all problems surrounding the bilateral trade setting, one basic open problem is to determine the worst-case approximation ratio of $\mathsf{SB}$ to $\mathsf{FB}$, i.e.
$$\inf_{F,G}\frac{\mathsf{SB}}{\mathsf{FB}}.$$
\end{problem}
\begin{problem}\label{prob:simple}
Seeing that the second best mechanism is usually too complex to be practical, it's natural to ask the same thing for simple mechanisms: Are there some simple BIC, IR and WBB mechanisms that achieve good approximations to the first best $\mathsf{GFT}$?
\end{problem}

It was not until the recent groundbreaking work of \cite*{DMSW:2022} that the second-best gains-from-trade was shown to provide a constant approximation to the first-best, i.e. $\inf_{F,G}(\mathsf{SB}/\mathsf{FB})>0$. In fact, they attack \Cref{prob:basic} by attacking \Cref{prob:simple} instead: they show that $\mathsf{FB}\leq 8.23\cdot\mathsf{RandOff}$. Since $\mathsf{SB}\geq \mathsf{RandOff}$, this implies that $\mathsf{FB}\leq 8.23\cdot\mathsf{SB}$. In other words, we can write $$\inf_{F,G}\frac{\mathsf{SB}}{\mathsf{FB}}\geq\inf_{F,G}\frac{\mathsf{RandOff}}{\mathsf{FB}}\geq \frac{1}{8.23}\approx 0.121.$$

In the hardness direction, \cite*{LLR:1989} and \cite{BM:2016} exhibit distributions for which the ratio $\mathsf{SB}/\mathsf{FB}$ is $2/e\approx 0.736$. \cite*{BDK:2021} show that $\inf_{F,G}(\mathsf{RandOff}/\mathsf{FB})\leq 0.495$. \cite*{CGMZ:2021} (in their arxiv version) independently give a proof that $\inf_{F,G}(\mathsf{RandOff}/\mathsf{FB})\leq \frac{1}{2}-\varepsilon$ for some constant $\varepsilon>0$.

This leaves a relatively large gap between the lower bound of $0.121$ and the upper bound of $0.736$ for \Cref{prob:basic}. For the approximation ratio of $\mathsf{RandOff}$, the gap is left at $[0.121, 0.495]$. In \Cref{sec:proof1}, we improve the lower bounds of both problems from $1/8.23\approx 0.121$ to $1/3.15\approx 0.317$, as another step towards the ultimate goal of closing the gaps.
\begin{theorem}
For any pair of distributions $F$ and $G$, the inequality $\mathsf{FB}\leq 3.15\cdot\mathsf{RandOff}$ holds.
\end{theorem}

In contrast to the constant approximation provided by the random-offerer mechanism, most of the other simple mechanisms, including the buyer-pricing mechanism, the seller-pricing mechanism and the fixed-price mechanism, cannot approximate $\mathsf{FB}$ in general. For the seller-pricing mechanism, it's well known that taking $G$ to be a single point mass at 0 and $F$ an equal-revenue-distribution makes the ratio $\mathsf{SellerP}/\mathsf{FB}$ tend to 0. Since $\mathsf{BuyerP}$ is symmetric to $\mathsf{SellerP}$, the buyer-pricing mechanism also has an approximation ratio of 0. The case of the fixed-price mechanism follows from taking $F$ and $G$ to be exponential distributions, as is shown by \cite{BD:2016}. Partly compensating for this hardness of approximation, there is another type of results concerned with the setting where $F$ and $G$ are subject to some restrictions. For example, by \citet{McAfee:2008} and \cite*{KV:2019}, 
$$\inf_{F=G}\frac{\mathsf{FixedP}}{\mathsf{FB}}=\frac{1}{2}.$$
As another example, \citet{BM:2016} consider a ``monotone hazard rate'' condition (see \Cref{def:MHR}) on $F$, and prove that
$$0.368\approx\frac{1}{e}\leq\inf_{\substack{F\in\mathcal{MHR}\\G}}\frac{\mathsf{SellerP}}{\mathsf{FB}}.$$
In \Cref{sec:proof2}, we show that the approximation ratio of $\mathsf{SellerP}$ to $\mathsf{FB}$ assuming monotone hazard rate of the buyer's distribution can be determined exactly (to be $1/(e-1)\approx 0.582$):
\begin{theorem}
If $F\in\mathcal{MHR}$, then $\mathsf{FB}\leq(e-1)\cdot\mathsf{SellerP}$, and the constant $(e-1)$ is optimal.
\end{theorem}
\subsection{More Related Work}
One might also ask about how much gains-from-trade would be lost if we use simple mechanisms instead of the known second-best mechanism. That is, the worst-case approximation ratios of simple mechanisms to $\mathsf{SB}$ are also of interest. \cite*{BCWZ:2017} show that
$\inf_{F,G}(\mathsf{RandOff}/\mathsf{SB})=\frac{1}{2}$.

There is also a line of research concerning the approximation of welfare, instead of gains-from-trade. Since we have
$$\mathsf{Welfare}=\E_{\substack{v\sim F\\ c\sim G}}\left[v\cdot x(v,c)+c\cdot\Big(1-x(v,c)\Big)\right]=\E_{c\sim G}[c]+\mathsf{GFT},$$
maximizing welfare is equivalent to maximizing gains from trade. In addition, providing a constant approximation to the first-best $\mathsf{GFT}$ also provides a constant approximation to the first-best welfare (but not vice versa). \citet{BD:2016} show that the fixed-price mechanism provides a $\left(1-\frac{1}{e}\right)$-approximation to the first-best welfare (improved to $1-\frac{1}{e}+0.0001$ by \cite*{KPV:2021}). \citet{BM:2016} show that the first-best welfare is inapproximable to above a fraction of 0.934 by BIC, IR and WBB mechanisms. 
\section{Bounding the Approximation Ratio}
\label{sec:proof1}
In this section, we will prove the main result
$\mathsf{FB}\leq 3.15\cdot\mathsf{RandOff}$.
As pointed out by \cite*{DMSW:2022}, as far as $\mathsf{SellerP},\mathsf{BuyerP}$ and $\mathsf{FB}$ are concerned, it is without loss of generality to assume that the distribution of $v$ and $c$ are supported on $[0,1]$ and have continuous and positive densities.
\subsection{Notational Preparations}
\begin{enumerate}
\item
The major advantage of assuming the existence of positive densities is that we can define the following ``quantile  function'', a powerful tool for analysis introduced by \cite*{DMSW:2022}: for each $x\in[0,1]$, define $\mu(x)$ to be the $(1-\lambda)$-quantile of $F|_{\geq x}$, i.e.
$$\mu(x)=F^{-1}\Big(1-\lambda+\lambda F(x)\Big),$$
where $\lambda\in(0,1)$ is a parameter to be chosen. Since $F$ is strictly increasing and continuous on $[0,1]$, so are $F^{-1}$ and $\mu$. Then, since $\mu$ is strictly increasing and continuous on [0,1], we can also define its inverse $\mu^{-1}$ on $[\mu(0),\mu(1)]=[\mu(0),1]$. 
\item
Let $\mathsf{SProfit}$ denote the maximum profit that the seller can gain in a seller-pricing mechanism, and let $\mathsf{BProfit}$ denote the maximum utility that the buyer can gain in a buyer-pricing mechanism. Obviously,
$\mathsf{SProfit}\leq\mathsf{SellerP}$ and $\mathsf{BProfit}\leq\mathsf{BuyerP}$. Since the quantities $\mathsf{SProfit}$ and $\mathsf{BProfit}$ are much easier to handle than $\mathsf{SellerP}$ and $\mathsf{BuyerP}$, we will prove the result $\mathsf{FB}\leq 3.15\cdot\left(\frac{1}{2}\mathsf{SellerP}+\frac{1}{2}\mathsf{BuyerP}\right)$ by showing $\mathsf{FB}\leq 3.15\cdot\left(\frac{1}{2}\mathsf{SProfit}+\frac{1}{2}\mathsf{BProfit}\right)$ instead.
\end{enumerate}

\subsection{Proof of Main Theorem}
The main theme in the proof is to put everything as an integration over the seller's cost $c$. Note that the definition of $\mathsf{SProfit}$ is already of this form, and $\mathsf{FB}$ can also easily be expressed in this form. \cite*{DMSW:2022} notice that $\mathsf{BProfit}$ can also be put into this form (at the expense of possibly shrinking it a little). We state it as follows:

\begin{lemma}\label{lemma:integration-over-c}
Let $\mu$ be the quantile function defined by any $\lambda\in(0,1)$. Then
$$\int_{0}^{1}\int_{\mu(c)}^{1}\Big(s-\mu^{-1}(s)\Big) \d F(s) \d G(c)\leq \mathsf{BProfit}.$$
\begin{proof}
By Fubini-Tonelli theorem, we can change the order of integration as follows:
\begin{align*}
&\int_{0}^{1}\int_{\mu(c)}^{1}\Big(s-\mu^{-1}(s)\Big)\d F(s)\d G(c)\\
=& \int_{\mu(0)}^{1}\int_{0}^{\mu^{-1}(v)}\Big(v-\mu^{-1}(v)\Big)\d G(c)\d F(v)\\
=& \int_{\mu(0)}^{1}\Big(v-\mu^{-1}(v)\Big)G(\mu^{-1}(v))\d F(v)\\
\leq& \int_{\mu(0)}^{1}\max_{p}(v-p)G(p)\d F(v)\\
\leq& \mathsf{BProfit}.\qedhere
\end{align*}
\end{proof}
\end{lemma}
In light of \Cref{lemma:integration-over-c}, we can now put all three quantities $\mathsf{FB},\mathsf{SProfit}$ and $\mathsf{BProfit}$ into integrations over $\d G(c)$. The integrands are:
\begin{align*}
\mathsf{FB}(c)&:=\int_{c}^{1}(v-c)\d F(v),\\
\mathsf{SProfit}(c)&:=\max_{p}(p-c)(1-F(p)),\\
\mathsf{BProfit}(c)&:=\int_{\mu(c)}^{1}\Big(s-\mu^{-1}(s)\Big) \d F(s).
\end{align*}
We have (the first two directly follow from the definition of $\mathsf{FB}$ and $\mathsf{SProfit}$, while the third follows from \Cref{lemma:integration-over-c})
\begin{align*}
\mathsf{FB}=&\int_{0}^{1}\mathsf{FB}(c)\d G(c),\\
\mathsf{SProfit}=&\int_{0}^{1}\mathsf{SProfit}(c)\d G(c),\\
\mathsf{BProfit}\geq&\int_{0}^{1}\mathsf{BProfit}(c)\d G(c).
\end{align*}
However, the expression of $\mathsf{BProfit}(c)$ looks nothing like $\mathsf{FB}(c)$ or $\mathsf{SProfit}(c)$. Our next step is to transform it into a more familiar form:  
\begin{lemma}\label{lemma:transforming}
Let $\mu$ be the quantile function defined by any $\lambda\in(0,1)$. Then for any $c\in[0,1]$,
$$\mathsf{BProfit}(c)=(1-\lambda)\cdot\mathsf{FB}(c)-\int_{c}^{\mu(c)}(v-c)\d F(v).$$
\end{lemma}
\begin{proof}
This follows from successive uses of Fubini-Tonelli theorem:
\begin{align*}
&\mathsf{BProfit}(c)\\=&\int_{\mu(c)}^{1}(s-\mu^{-1}(s))\d F(s)\\
=&\int\int_{\substack{\mu(c)\leq s\leq 1 \\ \mu^{-1}(s)\leq t\leq s}}1\d t\d F(s)\\
=&\int\int_{\substack{c\leq t\leq \mu(c)\\ \mu(c)\leq s\leq \mu(t)}}1\d t\d F(s)+\int\int_{\substack{\mu(c)<t<1\\ t\leq s\leq \mu(t)}}1\d t\d F(s)\\
=&\int_{c}^{\mu(c)}\Big(F(\mu(t))-F(\mu(c))\Big)\d t+
\int_{\mu(c)}^{1}\Big(F(\mu(t))-F(t)\Big)\d t\\
=&\int_{c}^{\mu(c)}\Big(F(\mu(t))-F(t)\Big)\d t-
\int_{c}^{\mu(c)}\Big(F(\mu(c))-F(t)\Big)\d t+
\int_{\mu(c)}^{1}\Big(F(\mu(t))-F(t)\Big)\d t\\
=&\int_{c}^{1}\Big(F(\mu(t))-F(t)\Big)\d t-\int_{c}^{\mu(c)}\Big(F(\mu(c))-F(t)\Big)\d t\\
=&\int_{c}^{1}(1-\lambda)\Big(1-F(t)\Big)\d t-\int_{c}^{\mu(c)}\Big(F(\mu(c))-F(t)\Big)\d t\\
=&(1-\lambda)\int_{c}^{1}\int_{t}^{1}1\d F(v)\d t-\int_{c}^{\mu(c)}\int_{t}^{\mu(c)}1\d F(v)\d t\\
=&(1-\lambda)\int_{c}^{1}\int_{c}^{v}1\d t\d F(v)-\int_{c}^{\mu(c)}\int_{c}^{v}1\d t\d F(v)\\
=&(1-\lambda)\int_{c}^{1}(v-c)\d F(v)-\int_{c}^{\mu(c)}(v-c)\d F(v)\\
=&(1-\lambda)\cdot\mathsf{FB}(c)-\int_{c}^{\mu(c)}(v-c)\d F(v).\qedhere
\end{align*}
\end{proof}
\begin{remark}
The work by \cite*{DMSW:2022} also contains a major part that transforms $\mathsf{BProfit}(c)$ into a nicer form. A key step in their transformation is to use summations to bound integrations, which incurs a loss in the constant. In comparison, our \Cref{lemma:transforming} only involves integrations and is completely lossless.
\end{remark}
Note that the minuend on the right hand side of the preceding lemma is already a fraction of $\mathsf{FB}(c)$. The next lemma shows that the diminution caused by the subtrahend is under control:   
\begin{lemma}\label{lemma:controlling}
Let $\mu$ be the quantile function defined by any $\lambda\in(0,1)$. Then for any $c\in[0,1]$,
$$\int_{c}^{\mu(c)}(v-c)\d F(v)\leq \ln\frac{1}{\lambda}\cdot\mathsf{SProfit}(c).$$
\end{lemma}
\begin{proof}
We resort to the definition of $\mathsf{SProfit}(c)$:
\begin{align*}
\int_{c}^{\mu(c)}(v-c)\d F(v)
&\leq\int_{c}^{\mu(c)}\frac{1}{1-F(v)}\cdot\max_{p}(p-c)(1-F(p))\d F(v)\\
&=\mathsf{SProfit}(c)\cdot\int_{c}^{\mu(c)}\frac{\d F(v)}{1-F(v)}\\
&=\mathsf{SProfit}(c)\cdot\ln\left(\frac{1-F(c)}{1-F(\mu(c))}\right)\\
&=\ln\frac{1}{\lambda}\cdot\mathsf{SProfit}(c)\qedhere
\end{align*}
\end{proof}
Now we can prove the theorem by combining the preceding three lemmas.
\begin{theorem}
$\mathsf{FB}\leq 3.15\cdot\mathsf{RandOff}$.
\end{theorem}
\begin{proof}
Let $\mu$ be the quantile function defined by any $\lambda\in(0,1)$. We have
\begin{align*}
\mathsf{FB}(c)=&\frac{1}{1-\lambda}\left((1-\lambda)\cdot\mathsf{FB}(c)-\int_{c}^{\mu(c)}(v-c)\d F(v)\right)+\frac{1}{1-\lambda}\int_{c}^{\mu(c)}(v-c)\d F(v)\\
\leq & \frac{1}{1-\lambda}\mathsf{BProfit}(c)+ \frac{1}{1-\lambda}\ln\frac{1}{\lambda}\cdot\mathsf{SProfit}(c)\quad(\text{By \Cref{lemma:transforming} and \Cref{lemma:controlling}}).
\end{align*}
Therefore,
\begin{align*}
\mathsf{FB}&=\int_{0}^{1}\mathsf{FB}(c)\d G(c)\\
&\leq \frac{1}{1-\lambda}\int_{0}^{1}\mathsf{BProfit}(c)\d G(c)+\frac{1}{1-\lambda}\ln\frac{1}{\lambda}\int_{0}^{1}\mathsf{SProfit}(c)\\
&\leq \frac{1}{1-\lambda}\mathsf{BProfit}+\frac{1}{1-\lambda}\ln\frac{1}{\lambda}\cdot\mathsf{SProfit}\\
&\leq \frac{1}{1-\lambda}\mathsf{BuyerP}+\frac{1}{1-\lambda}\ln\frac{1}{\lambda}\cdot\mathsf{SellerP}.
\end{align*}
By the symmetry described in \Cref{rem:symmetry}, we also have
$$\mathsf{FB}\leq\frac{1}{1-\lambda}\mathsf{SellerP}+\frac{1}{1-\lambda}\ln\frac{1}{\lambda}\cdot\mathsf{BuyerP}.$$
Adding up the preceding two inequalities, we get
\begin{align*}
\mathsf{FB}&\leq\left(\frac{1}{1-\lambda}+\frac{1}{1-\lambda}\ln\frac{1}{\lambda}\right)\left(\frac{1}{2}\mathsf{SellerP}+\frac{1}{2}\mathsf{BuyerP}\right)\\
&=\left(\frac{1}{1-\lambda}+\frac{1}{1-\lambda}\ln\frac{1}{\lambda}\right)\mathsf{RandOff}.
\end{align*}
Note that the above holds for all $\lambda\in(0,1)$. Thus the conclusion follows by calculating
\[\min_{0<\lambda<1}\left(\frac{1}{1-\lambda}+\frac{1}{1-\lambda}\ln\frac{1}{\lambda}\right)\approx 3.1462.\qedhere\]
\end{proof}
\section{Approximation Ratio under the MHR Condition}
\label{sec:proof2}
In \cref{sec:proof1}, we use the quantities $\mathsf{SProfit}$ and $\mathsf{BProfit}$ to lower-bound $\mathsf{SellerP}$ and $\mathsf{BuyerP}$, a (painful) compromise we make in the face of the difficulty in analyzing $\mathsf{SellerP}$ and $\mathsf{BuyerP}$ themselves. In this section, we will see that by imposing a restriction on the distribution of the buyer's value, one can significantly reduce the difficulty of analysis.
\subsection{Preliminaries}

We state the definition of the \textit{hazard rate} and the \textit{virtual value function}, which are commonly studied in auction theory since the work of \cite{Mye:1981}.  
\begin{definition}\label{def:MHR}
A distribution on $[0,1]$ with CDF $F$ and continuous and positive density function $f$ is said to have the \textit{monotone hazard rate} (MHR) property if the \textit{hazard rate} 
$$h(x)=\frac{f(x)}{1-F(x)}$$
is a monotone non-decreasing function of $x$.
\end{definition}

\begin{definition}
Define the \textit{virtual value function} of the buyer to be
$$\varphi(x)=x-\frac{1-F(x)}{f(x)}.$$
\end{definition}
It can be seen from this definition that the MHR property of $F$ implies that $\varphi$ is strictly increasing, and hence its inverse function $\varphi^{-1}$ exists on $[\varphi(0),1]\supset [0,1]$. What greatly reduces the difficulty of analysis is the following fact:
\begin{proposition}\label{prop:observe}
If the distribution of the buyer's value satisfies the MHR property, then the price $p_{c}$ that the seller would set given his cost $c$ is exactly $\varphi^{-1}(c)$. 
\end{proposition}
\begin{proof}
This is immediate from Myerson's auction theory (\cite{Mye:1981}). Here we give a short explanation, for the sake of completeness. Given $c\in [0,1]$, we have for each $p\in[c,1]$
$$
\frac{\d}{\d p}(p-c)(1-F(p))\\
=(1-F(p))-(p-c)f(p)\\
= f(p)(c-\varphi(p)).
$$
Since $\varphi$ is strictly increasing, the function
$p\mapsto (p-c)(1-F(p)) $
has a unique maximum $p^{*}=\varphi^{-1}(c)$.
\end{proof}
Let $\mathcal{MHR}$ be the collection of all MHR distributions on $[0,1]$. Our goal in this section is to show that
$$\inf_{\substack{F\in\mathcal{MHR}\\G}}\frac{\mathsf{SellerP}}{\mathsf{FB}}=\frac{1}{e-1}.$$

In \Cref{subsec:lower-bound}, we will prove that when $F\in\mathcal{MHR}$, the inequality $\mathsf{FB}\leq(e-1)\cdot\mathsf{SellerP}$ holds. Then, in \Cref{subsec:upper-bound} we will prove that the constant $(e-1)$ is optimal in the above inequality.
\begin{remark}
We can assume that the distribution defined by $G$ is supported on $[0,1]$ as well. Indeed, if we ``truncate'' $G$ into $G_{\text{truncate}}(x)=\begin{cases}
0 &\text{if } x<0\\
G(x) &\text{if } 0\leq x<1\\
1 &\text{if } x\geq 1
\end{cases}$, the ratio $\mathsf{SellerP}/\mathsf{FB}$ will not increase. Since we are concerned with the infimum of this ratio, this truncation is without loss of generality.
\end{remark}

\subsection{Proof of Lower Bound}
\label{subsec:lower-bound}
In light of \Cref{prop:observe}, we can define
$$\mathsf{SellerP}(c)=\int_{\varphi^{-1}(c)}^{1}(v-c)\d F(v),$$
and from the definition of $\mathsf{SellerP}$ we have
$$\mathsf{SellerP}=\int_{0}^{1}\mathsf{SellerP}(c)\d G(c).$$
Since we also have
$$\mathsf{FB}=\int_{0}^{1}\mathsf{FB}(c)\d G(c),$$
it suffices to show that $\mathsf{FB}(c)\leq (e-1)\cdot\mathsf{SellerP}(c)$ for each $c\in[0,1]$. Integrating by parts, we have
\begin{equation}\label{eq:1}
\begin{split}
\mathsf{FB}(c)&=\int_{c}^{1}(v-c)\d F(v)=(1-c)-\int_{c}^{1}F(v)\d (v-c)=\int_{c}^{1}(1-F(v))\d v\\
&=\int_{c}^{\varphi^{-1}(c)}(1-F(v))\d v+\int_{\varphi^{-1}(c)}^{1}(1-F(v))\d v
\end{split}
\end{equation}
and
\begin{equation}\label{eq:2}
\begin{split}
\mathsf{SellerP}(c)&=\int_{\varphi^{-1}(c)}^{1}(v-c)\d F(v)\\
&=(1-c)-\Big(\varphi^{-1}(c)-c\Big)F\left(\varphi^{-1}(c)\right)-\int_{\varphi^{-1}(c)}^{1}F(v)\d (v-c)\\
&=\int_{c}^{1}1\d v-\int_{c}^{\varphi^{-1}(c)}F\left(\varphi^{-1}(c)\right)\d v-\int_{\varphi^{-1}(c)}^{1}F(v)\d v\\
&=\int_{c}^{\varphi^{-1}(c)}\Big(1-F\left(\varphi^{-1}(c)\right)\Big)\d v+\int_{\varphi^{-1}(c)}^{1}\left(1-F(v)\right)\d v.
\end{split}
\end{equation}
Note that the right hand sides of $\cref{eq:2}$ and $\cref{eq:1}$ already have the second term in common. The next lemma relates their first terms also to each other.
\begin{lemma}\label{lemma:relates}
When $0\leq c\leq v\leq \varphi^{-1}(c),$ we have
$$\frac{1-F(v)}{1-F\left(\varphi^{-1}(c)\right)}\leq\exp\left(\frac{\varphi^{-1}(c)-v}{\varphi^{-1}(c)-c}\right).$$
\end{lemma}
\begin{proof}
By definition of the function $\varphi$, 
$$c=\varphi\left(\varphi^{-1}(c)\right)=\varphi^{-1}(c)-\frac{1}{h\left(\varphi^{-1}(c)\right)}.$$
Hence we get a nice expression for the hazard rate at $\varphi^{-1}(c)$:
$$h\left(\varphi^{-1}(c)\right)=\frac{1}{\varphi^{-1}(c)-c}.$$
Define a function $H(x)=-\ln (1-F(x))$ on $[0,1)$ (known as the \textit{cumulative hazard function}), we have for each $x\in[c,\varphi^{-1}(c)]$
$$H'(x)=\frac{f(x)}{1-F(x)}=h(x)\leq h\left(\varphi^{-1}(c)\right)=\frac{1}{\varphi^{-1}(c)-c}.$$
Integrating with respect to $x$ both sides of the above inequality from $v$ to $\varphi^{-1}(c)$, we get
$$H\left(\varphi^{-1}(c)\right)-H(v)\leq \frac{\varphi^{-1}(c)-v}{\varphi^{-1}(c)-c},$$
or
\[\frac{1-F(v)}{1-F\left(\varphi^{-1}(c)\right)}\leq\exp\left(\frac{\varphi^{-1}(c)-v}{\varphi^{-1}(c)-c}\right).\qedhere\]
\end{proof}
\begin{remark}
The proof in \cite{BM:2016} also contains a step that is equivalent to the special case $v=c$ of the preceding lemma, but they apply it in a different way to a different framework of analysis. 
\end{remark}
\begin{theorem}\label{thm:e-1}
When $F\in\mathcal{MHR}$, the inequality $\mathsf{FB}\leq(e-1)\cdot\mathsf{SellerP}$ holds.
\end{theorem}
\begin{proof}
By \Cref{lemma:relates}, for each $0\leq c<1$ and $c\leq v\leq\varphi^{-1}(c)$,
$$(1-F(v))\leq \exp\left(\frac{\varphi^{-1}(c)-v}{\varphi^{-1}(c)-c}\right)\cdot \Big(1-F\left(\varphi^{-1}(c)\right)\Big).$$
Integrating with respect to $v$ the above inequality from $c$ to $\varphi^{-1}(c)$, we get
\begin{align*}
&\int_{c}^{\varphi^{-1}(c)}(1-F(v))\d v\\\leq& \Big(1-F\left(\varphi^{-1}(c)\right)\Big)\cdot \int_{c}^{\varphi^{-1}(c)}\exp\left(\frac{\varphi^{-1}(c)-v}{\varphi^{-1}(c)-c}\right) \d v\\
=&\Big(1-F\left(\varphi^{-1}(c)\right)\Big)\cdot\Big(\varphi^{-1}(c)-c\Big)\cdot\int_{0}^{1}\exp(x)\d x \quad\left(\text{substituting }x=\frac{\varphi^{-1}(c)-v}{\varphi^{-1}(c)-c}\right)\\
=&(e-1)\cdot \int_{c}^{\varphi^{-1}(c)}\Big(1-F\left(\varphi^{-1}(c)\right)\Big)\d v.
\end{align*}
Therefore, combining the above with \cref{eq:1} and \cref{eq:2}, we have for each $c\in[0,1]$
\begin{align*}
\mathsf{FB}(c)&=\int_{c}^{\varphi^{-1}(c)}(1-F(v))\d v+\int_{\varphi^{-1}(c)}^{1}(1-F(v))\d v\\
&\leq (e-1)\int_{c}^{\varphi^{-1}(c)}\Big(1-F\left(\varphi^{-1}(c)\right)\Big)\d v+\int_{\varphi^{-1}(c)}^{1}\left(1-F(v)\right)\d v\\
&\leq (e-1)\int_{c}^{\varphi^{-1}(c)}\Big(1-F\left(\varphi^{-1}(c)\right)\Big)\d v+(e-1)\int_{\varphi^{-1}(c)}^{1}\left(1-F(v)\right)\d v\\
&=(e-1)\cdot\mathsf{SellerP}(c).
\end{align*}
It follows that $\mathsf{FB}\leq(e-1)\cdot\mathsf{SellerP}$.
\end{proof}
\subsection{Proof of Upper Bound}\label{subsec:upper-bound}
\begin{theorem}
The constant $(e-1)$ in \Cref{thm:e-1} is optimal.
\end{theorem}
\begin{proof}
We need only to construct examples of $F$ and $G$ such that the ratio $\mathsf{FB}/\mathsf{SellerP}$ can be infinitely close to $(e-1)$. Let $G$ be a single point mass at 0, and $F$ be the piecewise function:
$$
F(x)=\begin{cases}
1-e^{-x} & \text{ if } x\leq 1-\delta \\ 
1-e^{-1+\delta}(1-x)\left(\frac{1-\delta}{\delta^{2}}x-\frac{1-3\delta+\delta^{2}}{\delta^{2}}\right) & \text{ if } 1-\delta\leq x\leq 1 \\ 
\end{cases},
$$
where $\delta>0$ is very close to 0. The function $F$ is defined to be the CDF of an exponential distribution on $[0,1-\delta]$ and a quadratic function on $[1-\delta,1]$ that smoothly connects the exponential part and the point $F(1)=1$: 
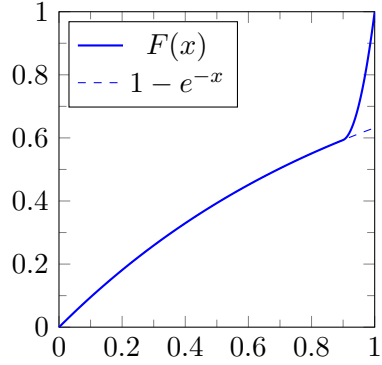
\begin{center}
\begin{tikzpicture}[
    declare function={
        func(\x)= (\x < 0.9)*(1-exp(-\x))   +
                (\x >= 0.9) * ((90*exp(-0.9)*\x-71*exp(-0.9))*(\x-1)+1)
     ;
  }
]
 
\begin{axis}[
    xmin = 0, xmax = 1,
    ymin = 0, ymax = 1,
    minor tick num = 1,
    width = 0.35\textwidth,
    height = 0.35\textwidth,
    legend pos=north west]
\addplot[
    domain = 0:1,
    samples = 400,
    smooth,
    thick,
    blue,
] {func(x)};
\addlegendentry{\(F(x)\)} 
\addplot[
    domain = 0:1,
    samples = 400,
    smooth,
    dashed,
    blue,
] {1-exp(-x)};
\addlegendentry{\(1-e^{-x}\)}
\end{axis}
\end{tikzpicture}
\end{center}
For $x\leq 1-\delta$, the CDF $F$ has a constant hazard rate $f(x)/(1-F(x))=e^{-x}/e^{-x}=1$, while on $[1-\delta,1]$, the function $F'=f$ is monotone increasing and hence the hazard rate of $F$ is monotone increasing on $[1-\delta,1]$. Thus $F$, as is defined above, satisfies the MHR property. For $x\leq 1-\delta$,
$$\varphi(x)=x-\frac{1-F(x)}{f(x)}=x-1<0,$$
and hence we must have $\varphi^{-1}(0)>1-\delta$. This implies that
$$\mathsf{SellerP}=\mathsf{SellerP}(0)=\int_{\varphi^{-1}(0)}^{1}(v-0)\d F(v)\leq \int_{1-\delta}^{1}v\d F(v)\leq \int_{1-\delta}^{1}1\d F(v)=F(1)-F(1-\delta)=e^{-1+\delta},$$
which tends to $e^{-1}$ as $\delta\rightarrow 0$. But we also have according to \cref{eq:1}
$$\mathsf{FB}=\mathsf{FB}(0)=\int_{0}^{1}(1-F(v))\d v,$$
which clearly tends to $\int_{0}^{1}e^{-x}\d x=1-e^{-1}$ when $\delta\rightarrow 0$.
This shows that
\[\inf_{\substack{F\in\mathcal{MHR}\\ G}}\frac{\mathsf{SellerP}}{\mathsf{FB}}\leq\frac{e^{-1}}{1-e^{-1}}=\frac{1}{e-1}.\qedhere\]
\end{proof}
\begin{remark}
Note that the hard case given above is actually when the seller has no cost. In this case, the gains-from-trade is equal to the welfare. Since any lower bound of the gains-from-trade approximation ratio always applies to the welfare approximation ratio as well, we can combine this observation with the lower bound proved in \Cref{subsec:lower-bound} to conclude that, assuming MHR of the buyer's distribution, the welfare approximation ratio of the seller-pricing mechanism to the first-best mechanism is also equal to $(e-1)$. 
\end{remark}
\begin{remark}
If we take the buyer's CDF $F$ to be the one in the preceding proof, and the seller's CDF to be $G(x)=1-F(1-x)$, it's easy to find that $\mathsf{RandOff}/\mathsf{FB}\rightarrow 1-1/e$ when the parameter $\delta\rightarrow 0$. Combining this with the lower bound in \Cref{subsec:lower-bound}, we have
$$0.582\approx\frac{1}{e-1}\leq\inf_{F,G^{\text{R}}\in\mathcal{MHR}}\frac{\mathsf{RandOff}}{\mathsf{FB}}\leq 1-\frac{1}{e}\approx 0.632,$$
where $G^{\text{R}}:=1-G(1-x)$. This shows that $\inf_{F,G^{\text{R}}\in\mathcal{MHR}}(\mathsf{RandOff}/\mathsf{FB})$ is strictly larger than $\inf_{F,G}(\mathsf{RandOff}/\mathsf{FB})$, which lies in $[0,317,0.495]$ (see \Cref{subsec:results}). 
\end{remark}

\section*{Acknowledgement}
The author would like to thank Kangning Wang and Zhaohua Chen for reading an earlier draft, discussion about the content, and their helpful suggestions on the presentation of the paper.

\bibliographystyle{plainnat}
\bibliography{reference}

\end{document}